\documentclass[12pt,english]{article}
\usepackage{amsmath}
\usepackage{amsfonts}
\usepackage{amssymb}
\usepackage{amsthm}

\usepackage[left=25mm,top=26mm,right=26mm,bottom=30mm]{geometry}

\newtheorem{theorem}{Theorem}[section]
\newtheorem{corollary}{Corollary}[section]
\newtheorem{lemma}{Lemma}[section]
\newtheorem{proposition}{Proposition}[section]

\let\epsilon=\varepsilon

\title{Uniformly Bounded Initial Chaos in Large System Often Intensifies
  Infinitely}

\author{A.A.~Lykov\thanks{Mechanics and Mathematics Faculty, Lomonosov Moscow
State University, Leninskie Gory~1, Moscow, 119991, Russia} \and  V.A.~Malyshev\footnotemark[1]}

\begin{document}

\maketitle

\begin{abstract}
We consider infinite harmonic chain with completely deterministic dynamics.
Initial data are assumed absolutely bounded. Nevertheless
 maximum of the variables can grow infinitely in time. We give conditions
 for this phenomenon. It coincides with intuitive guess that the main condition
 for this growth is sufficient chaos in the initial conditions.
\end{abstract}

\section{Introduction}

Our goal is to study models of various qualitative phenomena in non
equilibrium infinite particle systems. Normally such models use stochastic
dynamics. The goal of our project is to show that completely different
approach could be chosen -- with minimum possible probability. Here
we consider the simplest deterministic example of such models. We
assume that initially the deviations from equilibrium are uniformly
bounded. Could they grow in time and how? The same problem of course
exists for finite but large number of particles, but it demands some
scaling for time, number of particles etc. In recent papers  \cite{LM_4,LM_6,LM_7} 
we considered problems related to convergence to equilibrium for finite number of particles.

We consider trajectories $x_{k}(t),k\in Z,$ for standard countable
linear chain of particles defined by the formal interaction energy
\[
U=\frac{\omega_{0}^{2}}{2}\sum_{k}(x_{k}-a_{k})^{2}+\frac{\omega_{1}^{2}}{2}\sum_{k}(x_{k+1}-x_{k}-(a_{k+1}-a_{k}))^{2}
\]
where $\omega_{1}>0,$ and 
\[
\ldots <a_{k}<a_{k+1}<\ldots 
\]
where $a_{k}\to\pm\infty$ if $k\to\pm\infty$. Normally only the
case when $a_{k}=ka$ for some $a$, is considered but we will see
that it is almost the same. What is more important, we consider here
only the case when $\omega_{0}=0$, seemingly in case $\omega_{0}>0$
less chaos is expected as any particle is tightly bounded to a fixed place.
This case will be considered elsewhere.

If we introduce deviations $q_{k}=x_{k}-a_{k}$, then the energy can
be written as
\[
U=\frac{\omega_{1}^{2}}{2}\sum_{k}(q_{k+1}-q_{k})^{2}.
\]
The equilibrium corresponds to $q_{k}\equiv0$, and we consider the
equations 
\begin{equation}
\frac{d^{2}q_{k}}{dt^{2}}=\omega_{1}^{2}(q_{k+1}-2q_{k}+q_{k-1})=\omega_{1}^{2}(\Delta q)_{k}\label{infSystem}
\end{equation}
with some initial conditions $q_{k}(0),\dot{q}_{k}(0)$.

\section{Results}

Remind the following two spaces of functions on $\mathbb{Z}$: 
\[
l_{\infty}=l_{\infty}(\mathbb{Z})=\{f:\mathbb{Z}\to\mathbb{R}:\ \sup_{k\in\mathbb{Z}}|f(k)|<\infty\},\quad |f|_{\infty}=\sup_{k\in\mathbb{Z}}|f(k)|,
\]
\[
l_{2}=l_{2}(\mathbb{Z})=\Bigl\{f:\mathbb{Z}\to\mathbb{R}:\ \sum_{k\in\mathbb{Z}}|f(k)|^{2}<\infty\Bigr\},\quad |f|_{2}=\sqrt{\sum_{k\in Z}|f(k)|^{2}}.
\]
Put $p(t)=\dot{q}(t)$.

\begin{proposition} \label{l2bound} Assume $q(0)\in l_{2}(\mathbb{Z}),\ p(0)=0$,
then 
\[
|q(t)|_{\infty}\leqslant|q(0)|_{2}.
\]
\end{proposition}

Thus here the solution is uniformly  bounded (in particle's index $k$ and
time $t$). The situation drastically changes if the initial
conditions belong to $l_{\infty}$. The next statements show that
growth cannot exceed the order $\sqrt{t}$ and that there exist initial
conditions with this order of growth.

\begin{theorem} \label{thSolBounds} The following statements hold: 
\begin{enumerate}
\item Let $q(0)\in l_{\infty}(\mathbb{Z}),\ p(0)=0$, then for any $t\geqslant0$
the following inequality holds: 
\[
|q(t)|_{\infty}\leqslant\left(\sqrt{2\gamma\omega_{1}t}+2\right)|q(0)|_{\infty},
\]
where $\gamma>0$ is the unique positive root of the equation: 
\[
\frac{1}{\gamma}e^{1/ \gamma}=\frac{1}{e}.
\]
\item For any $n\in\mathbb{Z}$ there exist: a constant $c>0$, initial conditions
$q(0)\in l_{\infty}(\mathbb{Z}),\ p(0)=0$ and an increasing sequence
of time moments $t_{1}<t_{2}<\ldots,\ t_{k}\rightarrow\infty$ as
$k\rightarrow\infty$, such that 
\[
q_{n}(t_{2k})\geqslant c\sqrt{t_{2k}},\quad q_{n}(t_{2k+1})\leqslant-c\sqrt{t_{2k+1}}
\]
for any $k=1,2,\ldots$
\end{enumerate}
\end{theorem}

\begin{corollary} For any $n\in\mathbb{Z}$ there exist initial conditions
$q(0)\in l_{\infty}(\mathbb{Z}),$ $p(0)=0$ such that 
\[
\limsup_{t\rightarrow\infty}\frac{q_{n}(t)}{\sqrt{t}}=c_{1}>0,\quad\liminf_{t\rightarrow\infty}\frac{q_{n}(t)}{\sqrt{t}}=c_{2}<0
\]
for some constants $c_{1},c_{2}$. \end{corollary}

Next theorem claims that if the initial conditions are sufficiently
random, then any $q_{n}(t)$ can be arbitrary big with $t$.

\begin{theorem} \label{202002161400} Suppose that $q_{k}(0),\ k\in\mathbb{Z},$
is a  sequence of i.i.d.\ random variables such that $Eq_{k}(0)=0,Eq_{k}^{2}(0)>0$
and $E|q_{0}(0)|^{3}<\infty$. Then for all $n\in\mathbb{Z}$ the
following formula holds 
\begin{equation}
P\bigl\{ \sup_{t\geqslant0}q_{n}(t)=+\infty\bigr\} =P\bigl\{ \inf_{t\geqslant0}q_{n}(t)=-\infty\bigr\} =1,\label{202002092003}
\end{equation}
where by $P(\cdot)$ we denote the probability of the corresponding event.
\end{theorem}

Next theorem explains in some way the formal nature of the previous
results.

The following operator on $l_{\infty}$ 
\[
(Vq)_{k}=-\omega_{1}^{2}(\Delta q)_{k}=-\omega_{1}^{2}(q_{k+1}-2q_{k}+q_{k-1})
\]
is bounded, namely $|V|_{\infty}\leqslant4\omega_{1}^{2}$. Then the
following operator is also bounded on $l_{\infty}$: 
\begin{equation}
C(t)=\cos(t\sqrt{V})=\sum_{k=0}^{\infty}(-1)^{k}\frac{t^{2k}V^{k}}{(2k)!}.\label{cosDef}
\end{equation}

From (the proofs of) theorems \ref{thSolBounds} and \ref{existInitCond}
the following statement follows.

\begin{theorem} \label{thCosNormBounds} There exist constants $a,b>0$
such that for all $t\geqslant0$ the following inequalities hold:
\[
a\sqrt{t}+1\leqslant\left|\cos(t\sqrt{V})\right|_{\infty}\leqslant b\sqrt{t}+1.
\]
\end{theorem}

\section{Proofs}

\subsection{Various expressions for the dynamics}

\begin{lemma} The solution of the main system (\ref{infSystem})
can be written as 
\begin{equation}
q_{n}(t)=\sum_{k\in\mathbb{Z}}a_{k}(t)q_{n-k}(0)+\sum_{k\in\mathbb{Z}}b_{k}(t)p_{n-k}(0),\label{solFormula}
\end{equation}
where 
\begin{align*}
a_{k}(t)= & \ \frac{1}{2\pi}\int_{0}^{2\pi}\cos\left(2\omega_{1}t\sin\frac{\lambda}{2}\right)e^{ik\lambda}\ d\lambda,\\
b_{k}(t)= & \ \frac{1}{2\pi}\int_{0}^{2\pi}\frac{\sin\left(2\omega_{1}t\sin(\lambda / 2)\right)}{2\omega_{1}\sin(\lambda / 2)}e^{ik\lambda}\ d\lambda.
\end{align*}
\end{lemma}

\begin{proof} Let us prove that infinite vectors $a=\{a_{k}\}$ and $b$
satisfy the system (\ref{infSystem}). We have 
\[
\ddot{a}_{k}(t)=-\frac{\omega_{1}^{2}}{2\pi}\int_{0}^{2\pi} \Bigl(2\sin\frac{\lambda}{2}\Bigr)^{2}\cos\Bigl(2\omega_{1}t\sin\frac{\lambda}{2}\Bigr)e^{ik\lambda}\ d\lambda.
\]
From 
\[
2\sin^{2}\frac{\lambda}{2}=1-\cos\lambda=1-\frac{e^{i\lambda}+e^{-i\lambda}}{2}
\]
it follows 
\[
\ddot{a}_{k}(t)=-\omega_{1}^{2}(2a_{k}(t)-a_{k+1}(t)-a_{k-1}(t))=\omega_{1}^{2}(\Delta a)_{k}.
\]
Similarly for $b$. Moreover, 
\[
a_{k}(0)=\dot{b}_{k}(0)=\begin{cases}
1, & k=0,\\
0, & k\ne0.
\end{cases}
\]
It follows that the series in (\ref{solFormula}) formally satisfies
the system (\ref{infSystem}) with corresponding initial conditions.
It remains to prove that it is absolutely convergent and defines the
solution from $l_{\infty}$. Integrating by parts we get: 
\begin{align*}
  2\pi a_{k}(t)&=\frac{1}{ik}\cos\Bigl(2\omega_{1}t\sin\frac{\lambda}{2}\Bigr)e^{ik\lambda}\Big|_{0}^{2\pi}+\frac{1}{ik}\int_{0}^{2\pi}\omega_{1}t\cos\frac{\lambda}{2}
            \sin\Bigl(2\omega_{1}t\sin\frac{\lambda}{2}\Bigr)e^{ik\lambda}\ d\lambda
\\
        &=\frac{\omega_{1}t}{ik}\int_{0}^{2\pi}\cos\frac{\lambda}{2}\sin\Bigl(2\omega_{1}t\sin\frac{\lambda}{2}\Bigr)e^{ik\lambda}\ d\lambda\\
  &=\frac{\omega_{1}t}{k^{2}}\int_{0}^{2\pi}\frac{d}{d\lambda}\Bigl(\cos\frac{\lambda}{2}\sin\Bigl(2\omega_{1}t\sin\frac{\lambda}{2}\Bigr)\Bigr)e^{ik\lambda}\ d\lambda.
\end{align*}
It follows that $|a_{k}(t)|\leqslant c/k^{2}$ for some constant
$c$ not depending on $k$. Similar bound exists for $b_{k}$. Thus
the series in (\ref{solFormula}) are uniformly bounded by $c(t)|q(0)|_{\infty}$
for some function $c(t)$. The lemma is thus proved.
\end{proof}

We will need the following integral presentation for Bessel functions
of integer index: 
\[
J_{n}(t)=\frac{1}{\pi}\int_{0}^{\pi}\cos(t\sin\phi-n\phi)\,  d\phi.
\]
\begin{lemma} \label{convBesselSol} Let $p(0)=0$, then we have:
\begin{equation}
q_{n}(t)=\sum_{k}^{\infty}J_{2k}(2\omega_{1}t)q_{n-k}(0).\label{202003042303}
\end{equation}
\end{lemma}
\begin{proof} We use the equality 
\begin{align*}
  a_{k}(t)+a_{-k}(t)&=\frac{1}{\pi}\int_{0}^{2\pi}\cos\left(2\omega_{1}t\sin\frac{\lambda}{2}\right)\cos(k\lambda)\ d\lambda\\
  &=\frac{2}{\pi}\int_{0}^{\pi}\cos\left(2\omega_{1}t\sin\phi\right)\cos(2k\phi)\ d\phi
\\
                    &=\frac{1}{\pi}\int_{0}^{\pi}\left(\cos\left(2\omega_{1}t\sin\phi+2k\phi\right)+\cos\left(2\omega_{1}t\sin\phi-2k\phi\right)\right)d\phi\\
  &=J_{2k}(2\omega_{1}t)+J_{-2k}(2\omega_{1}t)=2J_{2k}(2\omega_{1}t).
\end{align*}
In the last equality we exploited the relation $J_{2k}(t)=J_{-2k}(t)$.
Using formula (\ref{solFormula}) and the fact that $a_{k}(t)-a_{-k}(t)$
is either zero or pure imaginary, we can get 
\begin{align*}
q_{n}(t)&=\sum_{k}\Bigl(\frac{1}{2}(a_{k}(t)+a_{-k}(t))+\frac{1}{2}(a_{k}(t)-a_{-k}(t))\Bigr) q_{n-k}(0)
\\
&=\sum_{k}\frac{1}{2}\left(a_{k}(t)+a_{-k}(t)\right)q_{n-k}(0)=\sum_{k}^{\infty}J_{2k}(2\omega_{1}t)q_{n-k}(0)
\end{align*}
Lemma \ref{convBesselSol} is proved. \end{proof}

Note that (\ref{202003042303}) is the Neumann series (see \cite{Watson})
for the solution $q_{n}(t)$ with coefficients determined by initial
conditions.

\begin{lemma} \label{cosSolPresen} The following formula holds:
\[
q(t)=C(t)q(0)+S(t)p(0),
\]
where $C(t)$ is defined in (\ref{cosDef}), and for any $t$ 
\[
S(t)=\sum_{k=0}^{\infty}(-1)^{k}\frac{t^{2k+1}}{(2k+1)!}V^{k}
\]
is the operator in $l_{\infty}$, continuous in $t$.
\end{lemma}
\begin{proof} It is clear that the power series in $t$ for $C(t)$ and $S(t)$
can be differentiated term by term, and we have the equalities: 
\begin{align*}
\ddot{C}(t)&=\sum_{k=1}^{\infty}(-1)^{k}\frac{t^{2k-2}V^{k}}{(2k-2)!}=-VC(t),\quad C(0)=E,
\\
\ddot{S}(t)&=\sum_{k=1}^{\infty}(-1)^{k}\frac{t^{2k-1}}{(2k-1)!}V^{k}=-VS(t),\quad S(0)=0,
\end{align*}
where $E$ is the unit operator.
\end{proof} 

\begin{lemma} \label{expBound} Let $q(0)\in l_{\infty},\, p(0)=0$,
and moreover $q_{n}(0)=0$ for all $|n|<M$ for some $M$. Then for
any $t\geqslant0$ the following inequality holds: 
\[
|q_{0}(t)|\leqslant\left(e^{\alpha+1}\alpha\right)^{2M}|q(0)|_{\infty},\quad \alpha=\omega_{1}\frac{t}{M}.
\]
\end{lemma}
\begin{proof}
  This statement easily follows from formula (\ref{202003042303})
and classical bounds for the Bessel functions (for example, from Poisson's
integral for the Bessel function, \cite{Watson}). But we will give
a proof without using a presentation of solution via the Neumann series.
So our proof can be easily generalized to more common quadratic interaction
potential.

We estimate $q(t)$ using lemma \ref{cosSolPresen}. We have: 
\[
q(t)=q(0)+\sum_{k=1}^{\infty}(-1)^{k}\frac{t^{2k}}{(2k)!}V^{k}q(0).
\]
For any $k\geqslant1$ the following holds: 
\[
(V^{k}q(0))_{0}=\sum_{i_{1},i_{2},\ldots,i_{k} \in\mathbb{Z}}V_{i_{0},i_{1}}V_{i_{1},i_{2}}\ldots V_{i_{k-1},i_{k}}q_{i_{k}}(0),\quad i_{0}=0.
\]
If $(V^{k}q(0))_{0}\ne0$ for some $k$, then one can find a sequence
$i_{1},\ldots,i_{k}$ of indices such that $|i_{r}-i_{r-1}|\leqslant1$
for all $r=1,\ldots,k$ and $|i_{k}| \geqslant M$. Thus 
\[
M \leqslant |i_{k}|=\biggl|\sum_{r=1}^{k}i_{r}-i_{r-1}\biggr|\leqslant\sum_{r=1}^{k}|i_{r}-i_{r-1}|\leqslant k.
\]
It follows that 
\[
q_{0}(t)=q_{0}(0)+\sum_{k=M}^{\infty}(-1)^{k}\frac{t^{2k}}{(2k)!}(V^{k}q(0))_{0}.
\]
Then for the operator norm we have 
\[
|V^{k}q(0)|_{\infty}\leqslant|V|_{\infty}^{k}|q(0)|_{\infty}\leqslant(4\omega_{1}^{2})^{k}|q(0)|_{\infty}.
\]
The following inequalities hold: 
\[
|q_{0}(t)|\leqslant\sum_{k=M}^{\infty}\frac{(2\omega_{1}t)^{2k}}{(2k)!}|q(0)|_{\infty}\leqslant e^{2\omega_{1}t}\frac{(2\omega_{1}t)^{2M}}{(2M)!}|q(0)|_{\infty}.
\]
Since for any integer $n$ 
\[
n!\geqslant\left(\frac{n}{e}\right)^{n},
\]
we have
\[
|q_{0}(t)|\leqslant\Bigl(\frac{\exp\{ \omega_{1}t / M \} \omega_{1}te}{M}\Bigr)^{2M}|q(0)|_{\infty}.
\]
The lemma is proved.
\end{proof}

\paragraph{Proof of Proposition \ref{l2bound}}

As the function $\cos\left(2\omega_{1}t\sin (\lambda / 2)\right)$
is infinitely smooth in $\lambda$ and periodic with period $2\pi$,
we have
\[
\cos\Bigl(2\omega_{1}t\sin\frac{\lambda}{2}\Bigr)=\sum_{k\in\mathbb{Z}}a_{k}(t)e^{-i\lambda k}
\]
for all $\lambda\in[0,2\pi]$ and all $t\geqslant0$. By formula (\ref{solFormula})
and Parseval's equality 
\begin{align*}
  |q_{n}(t)|&\leqslant\sqrt{\sum_{k}|a_{k}(t)|^{2}}\, \sqrt{\sum_{k}q_{k}^{2}(0)}\\
  &=|q(0)|_{2}\, \sqrt{\frac{1}{2\pi}\int_{0}^{2\pi}
    \cos^{2}\Bigl(2\omega_{1}t\sin\frac{\lambda}{2}\Bigr)\ d\lambda}\leqslant|q(0)|_{2}.
\end{align*}

\subsection{Proof of Theorem \ref{thSolBounds}}
\subsubsection{Upper bound }

Equivalent statement is that for all $t\geqslant0$ and all $n\in\mathbb{Z}$
the inequality 
$$|q_{n}(t)|\leqslant\left(\sqrt{2\gamma\omega_{1}t}+2\right)|q(0)|_{\infty}$$
holds. Without loss of generality we can put $n=0$. The idea of the
proof is the following. Write the initial condition as the sum of
two terms. The first one is the restriction of $q(0)$ on finite interval
of $\mathbb{Z}$ containing $0$ and having length of order $t$.
The second one is the restriction of $q(0)$ on the remaining part of
$\mathbb{Z}$. By Proposition \ref{l2bound} the solution corresponding
to the first term is of the order $\sqrt{t}$. It remains to show
that the influence of the second term for time of the order $t$ is
small. Now the formal proof.

Fix some time moment $T>0$ and a constant $\mu>0$. Write the initial vector
as: 
\[
q(0)=q^{2}(0)+q^{\infty}(0),
\]
where 
\begin{align*}
q_{n}^{2}(0)&=\begin{cases}
q_{n}(0), & |n|\leqslant\mu T,\\
0, & |n|>\mu T, 
\end{cases} \\
 q_{n}^{\infty}(0)&=\begin{cases}
0, & |n|\leqslant\mu T,\\
q_{n}(0), & |n|>\mu T.
\end{cases}
\end{align*}
Let $q^{2}(t),q^{\infty}(t)$ be the solutions of the system (\ref{infSystem})
corresponding to the initial conditions $q^{2}(0),p^{2}(0)=0$ and
$q^{\infty}(0),p^{\infty}(0)=0$. For all $t\geqslant0$: 
\[
q(t)=q^{2}(t)+q^{\infty}(t).
\]
Since $q^{2}(0)\in l_{2}$, by Proposition \ref{l2bound} the following
inequality holds for all $t\geqslant0$: 
\[
|q^{2}(t)|_{\infty}\leqslant|q^{2}(0)|_{2}\leqslant|q(0)|_{\infty}\sqrt{2\mu T+1}.
\]
By Lemma \ref{expBound} we have also the inequality: 
\begin{align*}
|q_{0}^{\infty}(t)|&\leqslant\left(e^{\alpha(t)+1}\alpha(t)\right)^{2m(T)}|q(0)|_{\infty},\\ 
\alpha(t)&=\omega_{1}\frac{t}{m(T)}
\end{align*}
where $m(T)=[\mu T]+1$, and $[x]$ is the integer part of $x$. Since
$\mu T / m(T) \leqslant1$ for all $T\geqslant0$, we have
\[
\alpha(T)\leqslant\frac{\omega_{1}}{\mu}.
\]
It follows 
\[
e^{\alpha(T)+1}\alpha(T)\leqslant e^{\omega_{1} / \mu} \frac{\omega_{1}}{\mu}e.
\]
Put $\mu=\gamma\omega_{1}$, where $\gamma$ satisfies the equation  
\[
e^{1 / \gamma} \frac{1}{\gamma}e=1
\]
(and thus it is uniquely defined).
For given $\mu$ the following inequality holds: 
\[
|q_{0}^{\infty}(T)|\leqslant|q(0)|_{\infty}.
\]
Thus, we have proved that for any $T\geqslant0$ and $\mu=\gamma\omega_{1}$,
the following inequality holds: 
\[
|q(T)|_{\infty}\leqslant(\sqrt{2\mu T+1}+1)|q(0)|_{\infty}\leqslant(\sqrt{2\gamma\omega_{1}T}+2)|q(0)|_{\infty}.
\]
The proof is finished.

\subsubsection{Lower bound}

Like in the proof of item 1, without loss of generality we take $n=0$.
If $q\in l_{\infty}$, then define the support of $q$ as follows:
\[
\mathrm{supp}(q)=\{n\in\mathbb{Z}:\ q_{n}\ne0\}.
\]
The idea of the proof is also simple. Firstly we prove that for any $T$
there exists initial vector having support of the order $T$, and
moreover such that maximum of the corresponding solution is bounded from below
by $c\sqrt{T}$ for some $c$, not depending on $T$. Then, for some
increasing sequence of $T_{k}$ we sum up these initial conditions
so that their supports do not intersect. Now the formal proof.

\begin{theorem} \label{existInitCond} There is $T_{0}$ such that
for any $T>T_{0}$ there are initial conditions $q(0)\in l_{\infty},\ |q(0)|_{\infty}\leqslant1,p(0)=0$
with the following properties: 
\begin{enumerate}
\item 
$
q_{0}(T)\geqslant c\sqrt{T}
$
for some constant $c>0$ not depending on $T$;
\item $\mathrm{supp}(q(0))\subset[aT,bT]$ for some positive constants $a,b>0$
not depending on $T$. 
\end{enumerate}
\end{theorem}
\begin{proof} For this proof we need some lemmas. For the
number $\mu\in\mathbb{R}$ define the functions: 
\[
J(t,\mu)=\frac{1}{\pi}\int_{0}^{\pi}\cos\bigl(t(\sin\phi-\mu\phi) \bigr)\  d\phi .
\]

\begin{lemma} For all $|\mu|<1$  the following
asymptotic formula holds as $t\rightarrow\infty$: 
\[
J(t,\mu)=\sqrt{\frac{2\pi}{t\sqrt{1-\mu^{2}}}}\cos\Bigl(tg(\mu)-\frac{\pi}{4}\Bigr)\Bigl(1+O\Bigl(\frac{1}{t}\Bigr)\Bigr)
\]
where $O()$ is uniform in $\mu$ on any segment $[a,b]\subset(-1,1)$
and 
\[
g(\mu)=\sqrt{1-\mu^{2}}-\mu\arccos(\mu).
\]
\end{lemma}
\begin{proof}
  This assertion is a well-known result about the
Bessel functions (see \cite{fedoruck,Watson}). Nevertheless we will
write the proof for completeness. We use stationary phase method for
the function $J(t,\mu)$ in case when the phase function contains
additional parameter (\hspace{-0.1pt}\cite{fedoruck}, p.\thinspace 107, theorem 1.6). In our
case the phase function is: 
\[
S(\phi,\mu)=\sin\phi-\mu\phi.
\]
We have the following equalities for the derivatives: 
\[
S_{\phi}(\phi,\mu)=\cos\phi-\mu,\quad S_{\phi\phi}(\phi,\mu)=-\sin\phi.
\]
It follows that $S$ has the unique critical point on $[0,\pi]$,
which for $|\mu|<1$ is equal to 
\[
\phi_{0}=\phi_{0}(\mu)=\arccos(\mu).
\]
Then 
\[
S_{\phi\phi}(\phi_{0}(\mu),\mu)=-\sqrt{1-\mu^{2}}.
\]
It follows that $\phi_{0}(\mu)$ is a non-degenerate stationary point
for $\mu\in[-1+\delta,1-\delta]$ for any $0<\delta<1$. Thus all
conditions of the mentioned theorem hold together with the following
formula: 
\[
J(t,\mu)=\frac{1}{\pi}\sqrt{\frac{2\pi}{t\sqrt{1-\mu^{2}}}}\cos\Bigl(tS(\phi_{0}(\mu),\mu)-\frac{\pi}{4}\Bigr)\Bigl(1+O\Bigl(\frac{1}{t}\Bigr)\Bigr)
\]
as $t\rightarrow\infty$, where $O()$ is uniform in $\mu\in[a,b]$
for any segment $[a,b]\subset(-1,1)$. We
have the equality: 
\[
S(\phi_{0}(\mu),\mu)=\sqrt{1-\mu^{2}}-\mu\arccos(\mu).
\]
So the lemma is proved.
\end{proof}

From the proven lemma it follows that, if $\mu t=k\in\mathbb{Z}>0$
for some $\mu$ such that $\mu< 1/2$, then: 
\[
J_{2k}(t)=J(t,2\mu)=\sqrt{\frac{2}{\pi t\sqrt{1-4\mu^{2}}}}\cos\Bigl(tg(2\mu)-\frac{\pi}{4}\Bigr)\Bigl(1+O\Bigl(\frac{1}{t}\Bigr)\Bigr),\;\, \mu=\frac{k}{t}
\]
as $t\rightarrow\infty$, where $O()$ is uniform in all $k$ belonging
to the segment $[at,bt]$ for all 
\[
0\leqslant a<b<\frac{1}{2}.
\]

Denote the main term in this asymptotic formula for $J_{2k}(t)$ by $f_{k}(t)$:
\[
f_{k}(t)=\sqrt{\frac{2}{\pi t\sqrt{1-4\mu^{2}}}}\cos\Bigl(tg(2\mu)-\frac{\pi}{4}\Bigr),\;\, \mu=\frac{k}{t}.
\]
Then 
\begin{equation}
J_{2k}(t)=f_{k}(t)\Bigl(1+O\Bigl(\frac{1}{t}\Bigr)\Bigr),\label{ckasym}
\end{equation}
where $O()$ has the properties as above.

\begin{lemma} \label{lowerBoundsumf} There exist numbers $0<a<b< 1/2$
and $t_{0}>0$ such that for any $t\geqslant t_{0}$ there is a subset
$I\subset([at,bt]\cap\mathbb{Z})$ with the property: 
\[
\sum_{k\in I}f_{k}(t)\geqslant c_{1}\sqrt{t},
\]
and moreover 
\[
|I|\leqslant c_{2}t
\]
for some positive constants $c_{1},c_{2}$ not depending on $t$.
\end{lemma}
\begin{proof} Denote 
\begin{equation}
x_{k}(t)=tg(2\mu)=tg(\nu_{k})=tg\Bigl(2\frac{k}{t}\Bigr),\;\, \nu_{k}=2\frac{k}{t}.\label{x_k}
\end{equation}
Then 
\[
f_{k}(t)=\sqrt{\frac{2}{\pi t\sqrt{1-4\mu^{2}}}}\cos\Bigl(x_{k}(t)-\frac{\pi}{4}\Bigr).
\]
To prove the lemma we need to examine the points $x_{k}(t)$ modulo
$2\pi$. It is sufficient to prove that the number of points $x_{k}$
in the interval $(0, \pi / 2)$ has order $t$.

We use the equality 
\[
g(\nu_{k+1})=g(\nu_{k})+\frac{2}{t}g'(\nu_{k})+\frac{4}{t^{2}}g''(\theta_{k}),
\]
for some $\theta_{k}\in[\nu_{k},\nu_{k+1}]$. Whence 
\[
x_{k+1}(t)=x_{k}(t)+2g'(\nu_{k})+\frac{4}{t}g''(\theta_{k}).
\]
Thus, $x_{k+1}$ and $x_{k}$ differ on some angle, which for large
$t$ equals approximately  $2g'(\nu_{k})$. Let us find the derivative
of $g$: 
\[
g'(\nu)=-\frac{\nu}{\sqrt{1-\nu^{2}}}-\arccos(\nu)+\frac{\nu}{\sqrt{1-\nu^{2}}}=-\arccos(\nu).
\]
Thus for $\nu\in(0,1)$, $g'(\nu)$ is negative and $g'(1)=0$. Fix
some small number $\epsilon$. It is clear that there exist interval
$(a,b)\subset(0,1)$ and $t_{0}$ such that for all $x,y\in(a,b)$
and all $t>t_{0}$ the following inequality holds: 
\[
-2\epsilon<2g'(x)+\frac{4}{t}g''(y)<-\epsilon.
\]
Then for all $k\in [at/2,bt/2]$ 
\begin{equation}
-2\epsilon<x_{k+1}(t)-x_{k}(t)<-\epsilon.\label{deltaXkineq}
\end{equation}

Consider the set 
\[
I=\Bigl\{k\in\Bigl[\frac{a}{2}t,\frac{b}{2}t\Bigr]:\ x_{k}(t)\in\Bigl(0,\frac{\pi}{2}\Bigr)  \mod 2\pi \Bigr\}.
\]
By (\ref{deltaXkineq}) we have: 
\[
c_{1}t\leqslant|I|\leqslant c_{2}t,
\]
for some positive constants $c_{1},c_{2}$ not depending on $t$.
Then, 
\[
\sum_{k\in I}f_{k}(t)\geqslant c_{1}\sqrt{\frac{2}{\pi t\sqrt{1-b^{2}}}}\frac{1}{\sqrt{2}}t=c_{3}\sqrt{t} ,
\]
and the lemma is proved.
\end{proof}

Now we come back to the proof of Theorem \ref{existInitCond}. Using
Lemma \ref{solFormula} we put 
\[
q_{k}=\begin{cases}
1, & k\in I,\\
0, & k\notin I .
\end{cases}
\]
Then by formula (\ref{ckasym}) and Lemma \ref{lowerBoundsumf} we
have 
\[
q_{0}\Bigl(\frac{t}{2\omega_{1}}\Bigr)=\sum_{k\in I}J_{2k}(t)=\frac{1}{2}\sum_{k\in I}f_{k}(t)\Bigl(1+O\Bigl(\frac{1}{t}\Bigr)\Bigr)\geqslant c_{1}\sqrt{t}+c_{2}\frac{1}{\sqrt{t}}.
\]
Theorem \ref{existInitCond} is thus proved.
\end{proof}

Now we prove part 2 of Theorem \ref{thSolBounds}. For the sequence
of time moments $T_{1}<T_{2}<\ldots<T_{k}\ldots$ denote by $q^{1}(t),q^{2}(t),\ldots,q^{k}(t),\ldots$
the solutions corresponding to those in the formulation of Theorem
\ref{existInitCond}. We shall assume that $T_{1}>T_{0}$ and 
\[
bT_{k}<aT_{k+1}
\]
for all $k=1,2,\ldots$. This inequality guaranties that 
\[
\mathrm{supp}(q^{i}(0))\cap\mathrm{supp}(q^{j}(0))=\emptyset
\]
for $i\ne j$, and we can define the sum: 
\[
q(0)=\sum_{i=1}^{\infty}q_{i}(0).
\]
Due to linearity for the solution with the initial condition $q(0),p(0)=0$
we have: 
\[
q(t)=\sum_{i=1}^{\infty}q_{i}(t).
\]
To estimate $q(T_{k})$ write 
\[
q(T_{k})=\sum_{i=1}^{k-1}q^{i}(T_{k})+q^{k}(T_{k})+\sum_{i=k+1}^{\infty}q^{i}(T_{k}).
\]
By Proposition \ref{l2bound} 
\[
|q_{0}^{i}(T_{k})|\leqslant\sqrt{(b-a)T_{i}}.
\]
Whence 
\[
\bigg |\sum_{i=1}^{k-1}q_{0}^{i}(T_{k}) \bigg | \leqslant\sqrt{b-a}\sum_{i=1}^{k-1}\sqrt{T_{i}}.
\]
 We estimate the third term above using Lemma \ref{expBound} and get:
\[
\bigg | \sum_{i=k+1}^{\infty}q_{0}^{i}(T_{k}) \bigg | \leqslant\left(e^{\alpha_{k}+1}\alpha_{k}\right)^{2M},\quad\alpha_{k}=\omega_{1}\frac{T_{k}}{M}
\]
where $M=[aT_{k+1}]+1$. Choose $T_{k+1}>T_{k}$ so that 
\[
e^{\alpha_{k}+1}\alpha_{k}<1.
\]
Then we have the estimate: 
\[
q_{0}(T_{k})\geqslant c\sqrt{T_{k}}-\sqrt{b-a}\sum_{i=1}^{k-1}\sqrt{T_{i}}-1.
\]
The final condition for $T_{k}$ is 
\[
\sqrt{b-a}\sum_{i=1}^{k-1}\sqrt{T_{i}}+1<\frac{c}{2}\sqrt{T_{k}} .
\]
With such choice of the sequence ${T_{k}}$ we have: 
\[
q_{0}(T_{k})\geqslant\frac{c}{2}\sqrt{T_{k}}.
\]
Thus the constructed sequence $T_{k}$ provides initial condition
and sequence ${t_{k}}$, satisfying the assertion of the second part
of Theorem \ref{thSolBounds}. Thus the theorem is proved.

\section{Proof of Theorem \ref{202002161400}}

Plan of the proof is the following. First we will prove that finite
dimensional distributions of $q_{0}(t+s)$ weakly converge as $t\rightarrow\infty$
to the finite dimensional distribution of some Gaussian stationary
random process for $s\in[0,+\infty]$. This fact allows us to prove
(\ref{202002092003}) in quite straightforward manner. Next without
loss of generality we will suppose that $\omega_{1}= 1/2$.

Define a family (parametrized by $t\geqslant0$) of processes
with smooth trajectories: 
\[
Q_{t}(s)=q_{0}(t+s),\ s\in[0,+\infty).
\]

Define a process $X(s)$ as a series: 
\begin{equation}
X(s)=\sum_{n\in\mathbb{Z}}\xi_{n}J_{n}(s),\label{Xsseries}
\end{equation}
where $\xi_{n}$ are independent standard Gaussian random variables.
Lemma \ref{additionalLemma} gives us: 
\[
\sum_{n\in\mathbb{Z}}J_{n}^{2}(s)=1.
\]
Therefore from Kolmogorov's two-series theorem follows the almost
sure convergence of series in (\ref{Xsseries}) for all $s\in\mathbb{R}$.
Obviously $X(s)$ is a zero mean Gaussian random process. Let us
calculate its covariance function again using lemma \ref{additionalLemma}:
\[
\mathrm{cov}(X(t),X(s))=\sum_{n\in\mathbb{Z}}J_{n}(s)J_{n}(t)=J_{0}(t-s).
\]
Therefore $X(s)$ is a stationary (in a wide sense) process with covariance
function 
\[
C_{X}(s)=\mathrm{cov}(X(s),X(0))=J_{0}(s).
\]

Lemma \ref{convergenceLemma} states that finite dimensional distribution
of $Q_{t}(s)$ converges as $t\rightarrow \infty$ to the finite dimensional
distribution of $ (\sigma^{2} / 2) X(s)$. Further on without loss
of generality we will assume that $\gamma= \sigma^{2} / 2 =1$.
The fact that maximum of $X(s)$ over $s\geqslant0$ is infinite almost
sure easily follows from the classical theory of stationary Gaussian
process (see \cite{CramerLeadbetter}). Therefore intuitively it 
is clear that maximum of $q_{0}(t)$ is infinite with probability
one. But we can not use this arguments in strong way while we will
not prove a weak convergence of the corresponding processes. We will
not follow this way. Instead of proving weak convergence we derive
(\ref{202002092003}) directly from Lemma \ref{convergenceLemma}.

We have the following equalities: 
\[
  P\Bigl\{ \sup_{t\geqslant0}q_{0}(t)=+\infty\Bigr\} =P\biggl\{ \bigcap_{a=1}^{+\infty}\{\sup_{t\geqslant0}q_{0}(t)\geqslant a\}\biggr\}
  =\lim_{a\rightarrow+\infty}P\Bigl\{ \sup_{t\geqslant0}q_{0}(t)\geqslant a\Bigr\} .
\]

Now we prove that 
\begin{equation}
P\Bigl\{ \sup_{t\geqslant0}q_{0}(t)\geqslant a\Bigr\} =1\label{202002171214}
\end{equation}
for all $a$.

Fix an arbitrary $\varepsilon>0$. Note that due to lemma \ref{besselUniformBound}
for all $\varepsilon'>0$ and all $N\geqslant1$ there exist $s_{1},\ldots_{N}$
such that 
\begin{equation}
|\mathrm{cov}(X(s_{i}),X(s_{j}))|=|C(s_{i}-s_{j})|\leqslant\varepsilon'\label{202002161514}
\end{equation}
for all $i\ne j$. In other words, $X(s_{1}),\ldots,X(s_{N})$ are
``almost'' independent. Indeed, put $s_{k}=k\delta$ where $\delta^{-1/3}<\varepsilon'$.
Then using lemma \ref{besselUniformBound} we obtain: 
\[
|\mathrm{cov}(X(s_{i}),X(s_{j}))|=|J_{0}(s_{i}-s_{j})|\leqslant|i-j|^{-1/3} \, \delta^{-1/3}\leqslant\varepsilon'.
\]

Further we will choose $\varepsilon'$ and $N$ explicitly and they
will depend on $\varepsilon$. Now let $s_{1},\ldots,s_{N}$ satisfy
(\ref{202002161514}). For all $T\geqslant0$ we have the bound 
\begin{equation}
P\Bigl\{ \sup_{t\geqslant0}q_{0}(t)\geqslant a\Bigr\} \geqslant P\Bigl\{ \sup_{k=1,\ldots,N}q_{0}(T+s_{k})\geqslant a\Bigr\} . \label{202002171212}
\end{equation}

From lemma \ref{convergenceLemma} it follows that there is $T_{0}\geqslant0$
such that for all $T\geqslant T_{0}$ the following inequality holds: 
\begin{equation}
\Big |P\Bigl\{ \sup_{k=1,\ldots,N}q_{0}(T+s_{k})\geqslant a\Bigr\} -P\Bigl\{ \sup_{k=1,\ldots,N}X(s_{k})\geqslant a\Bigr\} \Big | \leqslant\varepsilon.\label{202002171213}
\end{equation}
Now we want to use lemma \ref{202002171147} to estimate $P\left\{ \sup_{k=1,\ldots,N}X(s_{k})\geqslant a\right\} $.
Note that since $p(0)=\Phi(a)<1$ and $p(\delta)$ is an increasing continuous
function ($p(\delta)$ is defined in (\ref{202002171140})), there
 exist small $\delta'$ and a number $q<1$ such that for all $\delta<\delta'$
the following inequality holds: $p(\delta)<q$. Now suppose that $N\varepsilon'<\delta'$.
Lemma \ref{202002171147} gives us 
\[
P\Bigl\{ \sup_{k=1,\ldots,N}X(s_{k})\geqslant a\Bigr\} \geqslant1-q^{N}
\]
which is greater than $1-\varepsilon$ for sufficiently large $N$.
Using (\ref{202002171212}) and (\ref{202002171213}) we obtain 
\[
P\Bigl\{ \sup_{t\geqslant0}q_{0}(t)\geqslant a\Bigr\} \geqslant1-2\varepsilon.
\]
Since $\varepsilon$ is arbitrary, we have proved (\ref{202002171214}).
The proof for the $\inf$ is the same. This completes the proof of
the theorem.

\begin{lemma} \label{convergenceLemma} For all $s_{1},\ldots,s_{m}\geqslant0$
the following convergence of distributions holds: 
\[
\mathrm{Law}(Q_{t}(s_{1}),\ldots,Q_{t}(s_{m}))\rightarrow\mathrm{Law}(\gamma X(s_{1}),\ldots,\gamma X(s_{m})),\;\,  \gamma=\frac{\sigma^{2}}{2}
\]
as $t\rightarrow\infty$. \end{lemma} The proof is straightforward
and based on the continuity theorem for characteristic function. From
formula (\ref{202003042303}) we have: 
\[
q_{0}(t)=\sum_{n\in\mathbb{Z}}q_{n}(0)J_{2n}(t).
\]
Consider the characteristic function of the random vector $Q_{t}(s_{1}),\ldots,Q_{t}(s_{m})$:
\[
f_{t}(u_{1},\ldots,u_{m})=E\exp\biggl(i\sum_{j=1}^{m}u_{j}Q_{t}(s_{j})\biggr).
\]
Due to dominated convergence theorem and
independence of $q_{k}(0)$ we obtain: 
\[
f_{t}(u_{1},\ldots,u_{m})=E\exp{\biggl(i\sum_{n\in\mathbb{Z}}q_{n}(0)\sum_{j=1}^{m}u_{j}J_{2n}(t+s_{j})\biggr)}=\prod_{n\in\mathbb{Z}}h(\phi_{n}(t,\bar{u})),
\]
where $h(u)=E\exp(iq_{0}(0)u)$ is the characteristic function of
$q_{0}(0)$, and 
\[
\phi_{n}(t,\bar{u})=\sum_{j=1}^{m}u_{j}J_{2n}(t+s_{j}),
\]
with $\bar{u}=(u_{1},\ldots,u_{m})$. Fix $u_{1},\ldots,u_{m}$. Because
of lemma \ref{besselUniformBound} for all $\varepsilon>0$ we can
choose $t_{0}\geqslant0$ such that for all $t\geqslant t_{0}$ all
points $\phi_{n}(t,\bar{u}),\ n\in\mathbb{Z}$ lie in $\varepsilon$-neighborhood
of zero. And since $h(0)=1$, we can consider the principal branch
of logarithm: 
\begin{equation}
f_{t}(u_{1},\ldots,u_{m})=\exp\Bigl(\sum_{n\in\mathbb{Z}}\mathrm{Log}(h(\phi_{n}(t,\bar{u})))\Bigr).\label{charfutf}
\end{equation}
From smoothness it follows that for sufficiently small $u$ the following
formula holds 
\[
\mathrm{Log}(h(u))=-\frac{1}{2}\sigma^{2}u^{2}+O(u^{3}).
\]
Hence we have 
\begin{equation}
\mathrm{Log}(h(\phi_{n}(t,\bar{u})))=-\frac{1}{2}\sigma^{2}\phi_{n}^{2}(t,\bar{u})+O(\phi_{n}^{3}(t,\bar{u})).\label{Loghphuieq}
\end{equation}
Using lemma \ref{additionalLemma} we get 
\begin{align*}
\sum_{n\in\mathbb{Z}}\phi_{n}^{2}(t,\bar{u})&=\sum_{j,k=1,\ldots,m}u_{j}u_{k}\sum_{n\in\mathbb{Z}}J_{2n}(t+s_{j})J_{2n}(t+s_{k})\\
                                            &=\frac{1}{2}\sum_{j,k=1,\ldots,m}u_{j}u_{k}\left(J_{0}(2t+s_{j}+s_{k})+J_{0}(s_{j}-s_{k})\right)\\
  &\longrightarrow\frac{1}{2}\sum_{j,k=1,\ldots,m}u_{j}u_{k}J_{0}(s_{j}-s_{k})
\end{align*}
as $t\rightarrow\infty$. Next we study the sum of cubes. Note that
due to H\"older inequality we have an estimate: 
\[
|\phi_{n}^{3}(t,\bar{u})|\leqslant m^{2}\sum_{j=1}^{m}u_{j}^{3}J_{2n}^{3}(t+s_{j}).
\]
From this bound and lemma \ref{besselUniformBound} we obtain: 
\begin{align}
  \Bigl |\sum_{n\in\mathbb{Z}}\phi_{n}^{3}(t,\bar{u})\Bigr | &\leqslant m^{2}\sum_{j=1}^{m}u_{j}^{3}\sum_{n\in\mathbb{Z}}|J_{2n}^{3}(t+s_{j})|
                                                               \leqslant m^{2}t^{-1/3}\sum_{j=1}^{m}u_{j}^{3}\sum_{n\in\mathbb{Z}}J_{2n}^{2}(t+s_{j})
\nonumber \\
&=m^{2}t^{-1/3}\sum_{j=1}^{m}u_{j}^{3}\frac{J_{0}(2(t+s_{j}))+J_{0}(0)}{2} \longrightarrow 0 \label{mark_1)}
\end{align}
as $t\rightarrow\infty$. In equality (\ref{mark_1)}) we applied
lemma \ref{additionalLemma}.

Thus formulas (\ref{charfutf}) and (\ref{Loghphuieq}) give us: 
\[
\lim_{t\rightarrow\infty}f(t,u_{1},\ldots,u_{m})=\exp\Bigl(-\frac{\sigma^{2}}{4}\sum_{j,k=1,\ldots,m}u_{j}u_{k}J_{0}(s_{j}-s_{k})\Bigr).
\]
Applying the continuity theorem for characteristic functions we obtain
the assertion of lemma \ref{convergenceLemma}. This completes the
proof of the lemma.

\begin{lemma} \label{202002171147} Let $X_{1},\ldots,X_{N}$ be
a zero mean Gaussian random vector with $N\geqslant2$ such that:
\[
DX_{k}=1,\quad|\mathrm{cov}(X_{i},X_{j})|\leqslant\varepsilon',
\]
for all $k=1,\ldots,N$ and all $i\ne j$ and some $\varepsilon'>0$.
Assume the following inequality holds: 
\begin{equation}
N\varepsilon'\leqslant\delta<\frac{1}{2} . \label{202002161640}
\end{equation}
Then there is the estimate: 
\begin{equation}
P\Bigl\{ \sup_{k=1,\ldots,N}X_{k}\geqslant a\Bigr\} \geqslant1-p^{N},\quad p=p(\delta)=\sqrt{\frac{1+\delta}{1-\delta}}\,\Phi(a\sqrt{1+\delta})\label{202002171140}
\end{equation}
where $\Phi(x)$ is a cumulative distribution function of standard
normal distribution: 
\[
\Phi(x)=\frac{1}{\sqrt{2\pi}}\int_{-\infty}^{x}e^{-u^{2} / 2}du.
\]
\end{lemma} We have: 
\[
P\Bigl\{ \sup_{k=1,\ldots,N}X_{k}\geqslant a\Bigr\} =1-P\Bigl\{ \sup_{k=1,\ldots,N}X_{k}<a\Bigr\} .
\]
Denote by $C$  the covariance
matrix of the random vector $X$, $$C=(\mathrm{cov}(X_{i},X_{j}))_{i,j=1,\ldots,N}.$$ Due to Gershgorin circle theorem
all eigenvalues of $C$ lye in the circle with radius $(N-1)\varepsilon'< 1/2$
and center at $1$. Therefore $C$ is invertible and we can write
\[
P\Bigl\{ \sup_{k=1,\ldots,N}X_{k}<a\Bigr\} =\frac{1}{(\sqrt{2\pi})^{N}\sqrt{\det(C)}}\int_{[-\infty,a]^{N}}\exp\Bigl\{-\frac{(x,C^{-1}x)}{2}\Bigr\}dx,
\]
where by $(,)$ we denoted the standard Euclidean scalar product in
$\mathbb{R}^{N}$. We have an obvious inequality: 
\[
(x,C^{-1}x)\geqslant\frac{1}{\lambda}|x|^{2}
\]
where $|x|$ is standard Euclidean norm of the vector $x$ in $\mathbb{R}^{N}$
and $\lambda$ is the maximal eigenvalue of $C$. Hence we obtain
for the probability 
\begin{align}
  P\Bigl\{ \sup_{k=1,\ldots,N}X_{k}<a\Bigr\} &\leqslant\frac{1}{(\sqrt{2\pi})^{N}\sqrt{\det(C)}}\int_{[-\infty,a]^{N}}
                                               \exp\Bigl\{-\frac{|x|^{2}}{2\lambda}\Bigr\} dx \nonumber \\
  &=\frac{\sqrt{\lambda^{N}}}{\sqrt{\det(C)}}\Phi^{N}(\sqrt{\lambda}a).\label{202002171139}
\end{align}
Evidently we have $\det(C)=\lambda_{1}\ldots\lambda_{N}$ where $\lambda_{1},\ldots,\lambda_{N}$
are eigenvalues of $C$. Thus applying Gershgorin circle theorem we
get: 
\[
\det(C)\geqslant(1-\delta)^{N},\quad\lambda\leqslant1+\delta.
\]
Putting these inequalities to (\ref{202002171139}) we obtain (\ref{202002171140}).
This completes the proof of the lemma.

\begin{lemma} \label{besselUniformBound} For all $n\in\mathbb{Z}$
and every $t\in\mathbb{R}$ the following inequality holds 
\[
|J_{n}(t)| \leqslant\min\{|n|^{-1/3},|t|^{-1/3}\} .
\]
\end{lemma}

For the proof see \cite{Landau,Krasikov}.

\begin{lemma}[Neumann's additional theorem] \label{additionalLemma}
For all $\varphi,t_{1},t_{2}\in\mathbb{R}$ the following formula
holds : 
\begin{equation}
J_{0}(\bar{t})=\sum_{n\in\mathbb{Z}}J_{n}(t_{1})J_{n}(t_{2})\cos(n\varphi),\label{Neumann}
\end{equation}
where $\bar{t}=\sqrt{t_{1}^{2}+t_{2}^{2}-2t_{1}t_{2}\cos\varphi}$.
As a particular case of (\ref{Neumann}) we have: 
\begin{equation}
\sum_{n\in\mathbb{Z}}J_{2n}(t_{1})J_{2n}(t_{2})=\frac{1}{2}\left(J_{0}(t_{1}+t_{2})+J_{0}(t_{1}-t_{2})\right) . \label{additionalThN}
\end{equation}
\end{lemma}

One can find Neumann's additional theorem (\ref{Neumann})
in classical book \cite{Watson}, p.\thinspace 358--359. Equality (\ref{additionalThN})
immediately follows from (\ref{Neumann}) if we put $\varphi=0$ and
then $\varphi=\pi$ and next sum up two expressions. Here we used
that $J_{0}(t)=J_{0}(-t)$.

Let us make a remark about weak convergence of the process $q_{0}(t+s)$.
Due to additional theorem \cite[p.\thinspace 30]{Watson},  we have : 
\[
J_{2k}(t+s)=\sum_{n\in\mathbb{Z}}J_{n}(s)J_{2k-n}(t).
\]
Therefore due to formula (\ref{202003042303}) we obtain 
\[
q_{0}(t+s)=\sum_{n\in\mathbb{Z}}J_{n}(s)y_{n}(t),
\]
where 
\[
y_{n}(t)=\sum_{k\in\mathbb{Z}}q_{k}(0)J_{2k-n}(t).
\]
It is not hard to prove that $y_{n}(t)$ converges to $\xi_{n}$ in distribution
as $t\rightarrow\infty$, where $\xi_{n}$ are standard normal independent
random variables. At the next step one can prove the relative compactness
of the family $q_{0}(t+s),s\in[0,S]$ for some $S\geqslant0$, parametrized
by $t$. We have no need for weak convergence so we have omitted this
proof.

\end{document}